\documentclass[a4paper,10pt]{article}
\usepackage{amsmath,amssymb,amsthm}
\usepackage{algorithmic}
\usepackage{algorithm}
\usepackage{graphicx}
\usepackage{arydshln}
\usepackage{booktabs}
\newtheorem{corollary}{Corollary}

\newtheorem{lemma}{Lemma}

\newcommand{\poly}{\rm poly}

\newcommand{\Ber}{\rm Ber}
\newcommand{\F}{\mathbb{F}}
\newcommand{\N}{\mathbb{N}}

\title{Solving the LPN problem in cube-root time
  \footnote{Partially supported by ArmaSuisse funding ARAMIS R3210/047-12 and SNF grant No. 121874.} \\
}

\author{Urs Wagner\\
{\small {\em e-mail:\/} urs.wagner@math.uzh.ch \vspace{-1mm} }\\
{\small Mathematics Institute\vspace{-1mm}}\\
{\small University of Z\"urich\vspace{-1mm}}\\
{\small Winterthurerstr 190, CH-8057 Z\"urich, Switzerland }
}
\begin{document}
\maketitle

\begin{abstract}
In this paper it is shown that given a sufficient number of (noisy) random binary linear equations, 
the Learning from Parity with Noise (LPN) problem can be solved in essentially cube root time in the number of 
unknowns. 
The techniques used to recover the solution are known from fast correlation attacks on stream ciphers. 
As in fast correlation attacks, the performance of the algorithm depends on the number of equations given. It is
 shown that if this number exceeds a certain bound, and the bias of the noisy equations
is polynomial in number of unknowns $n$, the running time of the algorithm is reduced to 
$2^{\frac{n}{3}+o(n)}$ compared to the brute force checking of all $2^n$ possible solutions.
The mentioned bound is explicitly given and it is further shown that when this bound
 is exceeded, the complexity of the approach can even be further reduced.
\end{abstract}

\vspace{3mm}
\vspace{3mm}

\section{Introduction}
In many cryptanalyses, especially in fast correlation attacks on stream ciphers, 
some information on the secret key is leaked in form of a set of linear binary equations which are satisfied with probability 
bigger than one half. For each of these equations, let $q=\frac{1}{2}+\epsilon$ be the probability 
that the secret key is in the solution set. We call $\epsilon$
the bias. It is clear that if $\epsilon=\frac{1}{2}$, every equation essentially halfes the number of possible solutions,
as long as it is independent from the previous ones. 
In particular if the system of equation has full rank, the key can be recovered in polynomial time by simple Gaussian Elimination.
An interesting problem lies in how to recover the key if $\epsilon<\frac{1}{2}$. The LPN problem (e.g. \cite{ju05},\cite{le06}) 
captures the essence of this task.
Let
\begin{itemize}
 \item $x \in \F^n_2$ be a $n$-dimensional binary vector, also referred to as the key in the sequel.
 \item $E \sim \Ber(p)$ be a random variable with $\Pr(E=1)=p=\frac{1}{2}-\epsilon$ and 
$\Pr(E=0)=q=\frac{1}{2}+\epsilon$, $\epsilon \in \left[0,\frac{1}{2}\right]$.
 \item \O{}$_\epsilon$ be an oracle that uniformly at random chooses $g \in \F_2^n$ and outputs pairs 
$(\langle g,x \rangle + e,g)$ 
where $e$ is drawn according to $E$ and 
$\langle \cdot , \cdot \rangle$ denotes the usual inner product. The $g$'s can be seen as binary linear equations 
and computing the scalar product with $x$ corresponds to evaluating them at $x$.
\end{itemize}
The $n$-dimensional LPN$_\epsilon$ problem can be stated as follows: Given \O$_\epsilon${} and $\epsilon$, recover $x$.
A lower bound on the number $N$ of oracle calls necessary in order to be able to identify the correct $x$
with non-negligible probability can be given.
This bound corresponds to the number of samples $N$ necessary in order to make a good guess whether a random variable
$X$ is distributed according to Be$(p)$ or $X \sim$ Be$\left(\frac{1}{2}\right)$, where Be denotes the usual bernoulli distribution.
It is common knowledge that this number satisfies $N=O\left(\frac{1}{\epsilon^2}\right)$,
 as can easily be seen by Hoeffding's inequality \cite{bl03}.
As a consequence $x$ can be recovered in time $O(2^n \log \frac{1}{\epsilon^2})$ making 
$N=O\left(\frac{1}{\epsilon^2}\right)$ oracle calls. This is achieved by evaluating all equations 
 at $2^n$ points, using the techniques of fast Walsh transform \cite{ch02}.
While the LPN problem is proven to be NP-hard \cite{be78}, in the case where $N \gg \frac{1}{\epsilon^2}$
faster approaches than brute-force checking of all 
potential keys are possible. Especially the techniques known from fast correlation attacks 
(see e.g. \cite{ch02}, \cite{fo07}, \cite{lu04}) are well applicable. The core of most techniques lies in finding
linear combinations of the given equations such that a hypothesis on a subset of keybits can be tested. 
The application of these techniques to the LPN problem has been studied already in e.g. \cite{bl03, fo06, le06}.
While the attack we consider is not different to e.g. the one in \cite{fo06}, 
 the approach to the problem is another.
From past work, e.g. \cite{bl03,ch02,fo06,fo07,le06}, it is not immediately clear how the complexity
behaves depending on the number $N$ of random linear equations and the bias $\epsilon$. The influence of
$N$ and $\epsilon$ becomes explicit in our considerations. We will show that
if $\epsilon = \frac{1}{\poly(n)}$, then the LPN problem can be solved in time $2^{\frac{n}{3}+o(n)}$ given
 $N = 2^{\frac{n}{\log n}+o(n)}$ equations. 

The paper is organized as follows.
In Section \ref{sec:alg} a short overview on the fast correlation attack techniques is given.
In Section \ref{sec:main} it is shown how the complexity to recover the secret key depends on the number $N$ of oracle calls and 
the main result is stated at the end of the section. Section \ref{sec:evenmore} contains the case 
where the number $N$ of given equations
 exceeds the bound sufficient for a cube root attack.
In Section \ref{sec:conclusion} an illustrating example is given. 
Throughout the paper, $\log$ will denote the logarithm to base $2$.
\label{sec:introduction}

\section{Linear Combination and Hypothesis Testing}
\label{sec:alg}

Most fast correlation attacks rely on the principles of linear combination and hypothesis testing.
The goal of linear combination lies in constructing binary linear equations that
 depend only on a subset of the keybits. These
equations can then be used to test a hypothesis on this subset of keybits.
Let $g'_i=(\langle g_i,x \rangle + e_i, g_i) \in \F^{n+1}_2$ be a sample output by the oracle \O{}$_\epsilon$.
Note that if we add $w$ random samples $g'_{i_1},\dots,g'_{i_w}$ from \O{}$_{\epsilon}$, i.e.
if we consider 
$\tilde{g}=\left(\sum_j \langle g_{i_j},x \rangle + \sum_j e_{i_j}, \sum_j g_{i_j} \right)=\left(\langle \sum_j  g_{i_j},x \rangle + \sum_j e_{i_j}, \sum_j g_{i_j} \right) $
 this looks like a
sample output from \O{}$_{\tilde{\epsilon}}$, with $\tilde{\epsilon}=2^{w-1}\epsilon^{w}$. This can be seen 
by the well known Piling-up lemma (e.g. \cite{va06}).
By appropriately choosing $w$-tuples of samples from \O{}$_{\epsilon}$, we can get equations from 
\O{}$_{\tilde{\epsilon}}$ which depend only on a subset of keybits. 
\begin{lemma}
\label{lem:1}
Let $w \in \N$ be even and $N \gg w$ be the number of  samples given from \O{}$_\epsilon$.
Then all $w$-ary linear combinations of these equations which are all zero in the last 
$b$ bits can be found in time and space
\begin{equation}
O\left(\max \left\{N^{\frac{w}{2}},\frac{N^w}{2^{b}} \right\}\right).
\label{equ:compl_samples}
\end{equation}
\end{lemma}
\begin{proof}
The number of all $\frac{w}{2}$-ary linear combinations of the given equations equals
\begin{equation}
 {N \choose w/2}= \Theta\left(N^{\frac{w}{2}}\right),
\label{equ:w/2ary}
\end{equation}
for fixed $w$.
Compute these linear combinations and store the resulting equations in blocks according to the last $b$ bits,
i.e. inside a block the new equations coincide on the last $b$ bits.
In each of the $2^b$ blocks there are an expected number of 
\[
 \frac{{N \choose w/2}}{2^b}
\]
equations. Inside each block, take all $2$-ary combinations 
which every time gives an expected number of
\[
  \left(\frac{{N \choose w/2}}{2^b}\right)^2
\]
equations of the desired form. 
As there are $2^b$ blocks, we get an expected number
\begin{equation}
2^b\left(\frac{{N \choose w/2}}{2^b}\right)^2=\Theta\left(\frac{N^{w/2}}{2^b}\right),
\label{equ:2ary}
\end{equation}
equations of the desired form.
The complexity of the whole is the sum of the complexities for getting all the $\frac{w}{2}$-ary linear combinations 
(\ref{equ:w/2ary}) and
and all the $2$-ary combinations (\ref{equ:2ary}) inside the $2^b$ blocks.
\end{proof}
Clearly the $\frac{N^{w}}{2^{b}}$ equations found as in Lemma \ref{lem:1} depend on the first $n-b$ keybits only.
A hypothesis on these bits can be tested if 
\begin{equation}
\frac{N^{w}}{2^{b}} \geq \frac{1}{\epsilon'^2}= \frac{1}{2^{2(w-1)}\epsilon^{2w}},
\end{equation}
 as discussed in Section \ref{sec:introduction}. Note that we implicitly assume that the new equations are pairwise
independent, what seems to be a admissible assumption \cite{lu04}.
In order to find the correct $n-b$ keybits, all possible hypotheses on these bits are checked. This can be done by 
techniques of the fast Walsh transform \cite{ch02}.
\begin{lemma}
\label{lem:2}
Evaluating $\frac{1}{2^{2(w-1)}\epsilon^{2w}}$ binary linear equations in $n-b$ variables can be done in time
\begin{equation}
 2^{n-b} \log \frac{1}{2^{2(w-1)}\epsilon^{2w}}.
\label{equ:compl_evaluating}
\end{equation}
\end{lemma}
In the next section we derive the optimal choice of the parameters $w$ and $b$, and we will see how the 
resulting complexity behaves depending on $N$.

\section{Cube-root algorithm}
\label{sec:main}
Suppose we are given $N\geq \frac{1}{\epsilon^2}$ samples from \O{}$_\epsilon$.
In Section \ref{sec:alg} we have seen that
\begin{itemize}
 \item if for $w,b \in \N$ it holds that $w$ is even and  
\begin{equation}
\label{equ:numequ}
\frac{N^{w}}{2^{b}} \geq \frac{1}{2^{2(w-1)}\epsilon^{2w}},
\end{equation}
 then we can recover the first
$n-b$ keybits in time
\begin{equation}
 O\left(\max \left\{N^{\frac{w}{2}},\frac{N^w}{2^{b}}, 2^{n-b} \log \frac{1}{2^{2(w-1)}\epsilon^{2w}}\right\}\right).
\label{equ:complexity}
\end{equation}
\end{itemize}
In this section we will show how to find optimal parameters $b$ and $w$ such that the expression
 in (\ref{equ:complexity}) is minimal 
under the condition that the inequality (\ref{equ:numequ}) is satisfied.
Clearly (\ref{equ:numequ}) is equivalent to 
\[
w \left(\log N +2+2 \log \epsilon\right) \geq b+2,
\]
by taking the logarithm on both sides.
We will now show that in order to reach minimal complexity in (\ref{equ:complexity}) 
this inequality must be satisfied with equality.
Note that the right hand side of the inequality is increasing with $b$ and as $N \geq \frac{1}{\epsilon^2}$, 
the left hand side is increasing with $w$. 
Suppose that for a given choice of $b$ and $w$ the inequality is strict. 
Then either $b$ can be increased or $w$ can be decreased resulting in a decrease of the overall complexity 
(\ref{equ:complexity}), while the inequality still holds.
So we can require equality 
\begin{equation}
w \left(\log N +2+2 \log \epsilon\right)-2=b.
\label{equ:b}
\end{equation}
Using this in equation (\ref{equ:complexity}),  we get the following overall complexity
\[
O\left(\max \left\{N^{\frac{w}{2}},\frac{1}{2^{2(w-1)}\epsilon^{2w}}, 2^{n}\frac{1}{N^w} \frac{1}{2^{2(w-1)}\epsilon^{2w}}\log \frac{1}{2^{2(w-1)}\epsilon^{2w}}\right\}\right).
\]
In order to ease discussion we adjust the condition on $N$.
From now on we will assume that
\[
 N \geq \frac{4}{(2\epsilon)^4}.
\]
As a direct consequence
\[
 N^{\frac{w}{2}} \geq \frac{1}{2^{2(w-1)}\epsilon^{2w}},
\]
and the overall complexity equals
\[
O\left(\max \left\{\underbrace{N^{\frac{w}{2}}}_{\alpha(w)}, \underbrace{2^{n}\frac{1}{N^w} \frac{1}{2^{2(w-1)}\epsilon^{2w}}\log \frac{1}{2^{2(w-1)}\epsilon^{2w}}}_{\beta(w)}\right\}\right).
\]
One readily verifies that $\alpha(w)$ is growing with $w$ and $\beta(w)$ is decreasing with $w$.
Hence the whole term reaches its minimum at the intersection of the two functions, i.e. if $\alpha(w)=\beta(w)$.
In order to get an (approximate) solution for the equation $\alpha(w)=\beta(w)$, we ignore the logarithmic term in $\beta(w)$
and obtain:
\begin{equation}
\label{equ:w}
 w=\frac{n+2}{3/2 \log N +2+2\log \epsilon}.
\end{equation}
Using (\ref{equ:b}), for $b$ we obtain:
\begin{equation}
 b=w \left(\log N +2+2 \log \epsilon\right)-2=n-\frac{(n+2)\log N}{3\log N+4+4\log \epsilon}=n-\frac{w}{2}\log N.
\label{equ:bb}
\end{equation}
We will now examine how this choice of the parameters affects the complexity of the linear combination and hypothesis testing 
approach.
For simplicity in the further analysis let us define
\begin{equation}
T_\epsilon(N):= \frac{\log N}{3\log N+4+4\log \epsilon}.
\label{equ:T_e}
\end{equation}
So we can write
\begin{equation}
 w=\frac{2(n+2)}{\log N}T_\epsilon(N),
\label{equ:w_T}
\end{equation}
and
\begin{equation}
 b=w\left(\log N +2+2 \log \epsilon \right)-2=n-(n+2)T_\epsilon(N).
\label{equ:b_T}
\end{equation}
\begin{lemma}
\label{lem:OO}
Notation as in the considerations before. Making  $N \geq \frac{4}{(2\epsilon)^4}$ oracle calls
and writing 
 $r:= w-2 \left\lfloor \frac{w+1}{2} \right\rfloor$ and 
 $T:= T_\epsilon(N)$,
the $n$-dimensional LPN$_\epsilon$ problem can be solved in time and space
\begin{equation}
O\left(2^{(n+2)T+ |r| \log N + \log \left((n+2)T+ \log N\right) }\right).
\label{equ:OO}
\end{equation}
\end{lemma}
\begin{proof}
Let $w$ be as in (\ref{equ:w}). Define
\begin{equation}
\label{equ:w'b'}
 w':=2 \left\lfloor \frac{w+1}{2} \right\rfloor \mbox{ and } b':=\left\lfloor w'\left(\log N +2+2 \log \epsilon \right)-2\right\rfloor.
\end{equation}
This definition ensures that $w'$ and $b'$ are integers and $w'=w+r$ is even with $r\in [-1,1]$.
Further
\[
 \frac{N^{w'}}{2^{b'}} \geq  \frac{N^{w}}{2^{b}} = \frac{1}{2^{2(w-1)}\epsilon^{2w}},
\]
so we have enough equations to check a hypothesis on the $n-b'$ nonzero bits.
The complexity for finding the $w'$-ary linear equations equals
\[
N^\frac{w'}{2}=N^{\frac{w}{2}+\frac{r}{2}} = 2^{(n+2)T +\frac{r}{2} \log N}.
\]
Let us now examine the complexity for evaluating these equations at $2^{n-b'}$ points.
We have
\begin{eqnarray*}
 b'&=& \left\lfloor w'\left(\log N +2+2 \log \epsilon \right)-2\right\rfloor 
\\ &=&\left\lfloor (w+r)\left(\log N +2+2 \log \epsilon \right)-2\right\rfloor
\\ &=&\left\lfloor w\left(\log N +2+2 \log \epsilon \right)-2 +r\left(\log N +2+2 \log \epsilon \right)\right\rfloor
\\ &\geq& w\left(\log N +2+2 \log \epsilon \right)-2 +r\left(\log N +2+2 \log \epsilon \right) -1/2
\\ &\stackrel{(\ref{equ:b_T})}{=}&  n-(n+2)T +r\left(\log N +2+2 \log \epsilon \right) -1/2.
\end{eqnarray*}
Hence 
\[
  2^{n-b'} \log \frac{N^{w'}}{2^{b'}}  
\leq 2^{(n+2)T-r(\log N +2 +2\log \epsilon )+1/2} \underbrace{\log N^\frac{w'}{2}}_{\leq (n+2)T +\frac{1}{2} \log N}.
\]
Adding these two upper bounds, we obtain the overall complexity
\begin{eqnarray}
  2^{n-b'} \log \frac{N^{w'}}{2^{b'}}+N^\frac{w'}{2} 
&\leq& 2^{(n+2)T}\left(2^{-r(\log N +2+2 \log \epsilon)+\frac{1}{2}+ \log ((n+2)T +\frac{r}{2} \log N)}+2^{\frac{r}{2}\log N}\right)  
\\&=&O\left(2^{(n+2)T+ r\log N + \log \left((n+2)T+\frac{1}{2} \log N\right) }\right)
\end{eqnarray}
\end{proof}
\begin{corollary}
\label{cor:OO}
Using the notation from the previous lemma. Making  $N \geq \frac{4}{(2\epsilon)^4}$ oracle calls
the $n$-dimensional LPN$_\epsilon$ problem can be solved in time and space
\begin{equation*}
O\left(2^{(n+2)T+ \log N + \log \left((n+2)T+ \log N\right) }\right).
\end{equation*}
\end{corollary}
\begin{proof}
Immediately as $|r|=\left| w-2 \left\lfloor \frac{w+1}{2} \right\rfloor \right| \leq 1$.
\end{proof}
It is not hard to see that if $\log N$ is large, $T_\epsilon(N)$ will converge to $\frac{1}{3}$. Clearly
\begin{equation*}
 T_\epsilon(N)= \frac{\log N}{3\log N+4+4\log \epsilon} =
 \frac{1}{3}\left(1-\frac{4+4 \log \epsilon}{3\log N + 4+4 \log \epsilon}\right).
\end{equation*}
Recall that $\log \epsilon<-1$ and since $N>\frac{4}{(2\epsilon)^4}$ we have that $\log N > -2 - 4\log\epsilon$.
Consequently
\[
  T_\epsilon(N) < \frac{1}{3}\left(1+\frac{4+4 \log \epsilon}{2+8\log \epsilon}\right) < \frac{1}{3}\left(1+\frac{1}{2}\right)=\frac{1}{2}.
\]
In the case where $N$ is significantly bigger, particularly if 
\[
 N \geq \frac{4}{(2\epsilon)^4} 2^{\frac{n}{\log n}}>2^{\frac{n}{\log n} -\frac{4}{3} -\frac{4}{3} \log \epsilon},
\]
one readily verifies that
\begin{equation}
\label{equ:T}
 T_\epsilon(N)\leq \frac{1}{3}\left(1-\frac{4}{3}\frac{(1+\log \epsilon)\log n}{n}\right).
\end{equation}
We can prove the following lemma:
\begin{lemma}
If $\epsilon = \frac{1}{\poly(n)}$ we can solve the $n$-dimensional LPN$_\epsilon$
in time and space
\begin{equation*}
2^{\frac{n}{3} + o(n)},
\end{equation*}
making  $N \geq 2^{\frac{n}{\log n}} \frac{4}{(2\epsilon)^4}$ oracle calls.
\end{lemma}
\begin{proof}
First notice that with $\epsilon = \frac{1}{\poly(n)}$, we have that $\log \frac{1}{\epsilon} =  o(n)$.
We will use only a subset of $N'= 2^{\frac{n}{\log n}} \frac{4}{(2\epsilon)^4}$ of the given equations.
Write $T':=T_\epsilon(N')$.
Then 
\[
 (n+2)T'= \frac{n+2}{3}\left(1-\frac{4\log n(\log \epsilon +1)}{3n}\right) = \frac{1}{3}n
 \underbrace{+\frac{2}{3}-\frac{4(n+2)}{n}(\log \epsilon +1)\log n}_{=o(n)}.
\]
Further
\[
 \log N' = \frac{n}{\log n} + \log \frac{4}{(2 \epsilon)^4} = o(n).
\]
So from Lemma \ref{lem:OO} and as also $\log \left((n+2)T'+ \log N'\right) = o(n)$, we get that we can find the solution in time
\[
O\left(2^{(n+2)T'+ \log N' + \log \left((n+2)T'+ \log N'\right) }\right)= 2^{\frac{n}{3}+o(n)}.
\]
\end{proof}

We have seen that we do not need more than  $N=2^{\frac{n}{\log n}+o(n)}$ equations to solve the LPN problem in essentially 
cube-root time. As seen in the proof of Lemma 4, given $N \gg 2^{\frac{n}{\log n}}$ the approach makes use of
$N'= 2^{\frac{n}{\log n}} \frac{4}{(2\epsilon)^4}$ of the given equations. 
The resulting overhead can be exploited to further reduce the complexity. The principle used
 in the case where $N \gg 2^{\frac{n}{\log n}}$ is called decimation \cite{fo06}. Given
\[
 N=2^l 2^{\frac{n}{\log n}} \frac{4}{(2\epsilon)^4}
\]
equations, the problem is basically reduced to solving the LPN problem in dimension $n-l$.

\subsection{Decimation}
\label{sec:evenmore}
We have seen that if $N = 2^{\frac{n}{\log n}}\frac{4}{(2\epsilon)^4}$ the complexity of the LPN problem is  $2^{\frac{n}{3}+o(n)}$.
If we are given $N \gg 2^{\frac{n}{\log n}}\frac{4}{(2\epsilon)^4}$ equations, simple decimation allows to reduce the 
security parameter $n$ of the problem.
Suppose we are given 
\[
N= 2^{\frac{n}{\log n}+l}\frac{4}{(2\epsilon)^4},
\]
 equations with $l<n-2$.
We want to consider only the equations that do not depent on (e.g. the first) $l' \in \N$ bits of the key $x$.
We have an expected number $N2^{-l'}$ of such equations.
In order to be able to recover the remaining $n-l'$ keybits, the following equality must hold
\[
 N 2^{-l'}=2^{\frac{n}{\log n}+l-l'}\frac{4}{(2\epsilon)^4} \geq 2^\frac{n-l'}{\log (n-l')}\frac{4}{(2\epsilon)^4}.
\]
Equivalently, 
\[
 \frac{n}{\log n}+l-l' \geq \frac{n-l'}{\log (n-l')}.
\]
Setting $l'=\lfloor l \rfloor \leq n-3$, this inequality is satisfied
 and we can reduce the problem parameter $n$ to $n-l$.
\begin{lemma}
If $\epsilon = \frac{1}{\poly(n)}$ we can solve the LPN$_\epsilon$
in time and space
\[
 2^{\frac{n-l}{3} + o(n)}.
\]
making $N \geq 2^{\frac{n}{\log n}+l} \frac{4}{(2\epsilon)^4}$ oracle calls.
\end{lemma}

\section{Example}
\label{sec:conclusion}

We have seen that the LPN problem can be solved in essentially cube-root time and space.
Consider the classical setting of a fast correlation attack \cite{me88}.
Suppose we have a stream cipher with keylength $n=128$ whose output bits correspond to linear 
combinations of the keybits transmitted over the binary symmetric channel with crossover 
probability $\frac{1}{2}-\epsilon=\frac{1}{2}-\frac{1}{8}=0.375$. 
For a given number $N\geq 2^{10}$ (note that $\log \frac{4}{(2\epsilon)^4}=10$) of equations, 
we have seen how to in principle optimally choose
$w$ and $b$ (see  (\ref{equ:w}) and (\ref{equ:bb})).
However these values are not necessarily in $\N$ and $w$ is not necessarily even.
So $w$ is rounded to the nearest even number $w'$ and $b'$ is chosen accordingly (see  (\ref{equ:w'b'}) in the proof of 
Lemma \ref{lem:OO}). This gives an additional summand $< |r| \log N$ in the exponent of the complexity (compare (\ref{equ:OO})).
Table \ref{table:overview} shows how this rounding problem influences the complexity.
 Decimation is not considered in
this example.

\begin{table}[h]
\renewcommand{\arraystretch}{1.3}
\begin{center}
\begin{tabular}{  l l l  l l l l  l } \toprule
$\log N$ ~\	& $w$  ~\	& $b$  ~\ 	& $w'$ ~\ &  $b'$ ~\ 	& $|r| \log N$ ~\ & $\log C_{LC}$ ~\ & $\log C_{HT}$ \\  \midrule
10	  	&  11.82	&  68.91	& 12	&  70 		& 	1.8	&	60	&  63.64    \\  
20	  	&  5		&  78	 	& 6	&  94 		& 	20	&	60	&  38.70    \\  
30	  	&  3.17		&  80.44	& 4	&  102 		& 	24.8	&	60	&  30.17    \\  
40	  	&  2.32		&  81.57	& 2	&  70 		& 	12.8	&	40	&  61.32   \\  
47	  	&  1.95		&  82.06	& 2	&  84 		& 	2.1	&	47	&  47.32   \\ 
50  		&  1.83		&  82.22	& 2	&  90		& 	8.5	&	50	&  41.32  \\ \toprule
\end{tabular}
\end{center}
\label{table:overview}
\caption{Complexity of Linear Combination $C_{LC}$ and Hypothesis testing $C_{HT}$ depending on the number of equations.}
\end{table}

\section*{Acknowledgment}

The author would like to thank G. Maze for many useful comments and discussions.

\bibliographystyle{plain}


\end{document}